\documentclass{ifacconf}

\usepackage{graphicx}  
\usepackage{natbib}   
\usepackage{xcolor}
\usepackage{amsmath, amssymb}
\usepackage{mathtools}
\usepackage{algorithm}
\usepackage[noend]{algpseudocode}

\algrenewcommand{\Require}{\Statex \textbf{Input:}}  
\algrenewcommand{\Ensure}{\Statex \textbf{Output:}} 

\newtheorem{definition}{Definition}
\newtheorem{assumption}{Assumption}
\newtheorem{remark}{Remark}
\newtheorem{problem}{Problem}
\newtheorem{example}{Example}

\begin{document}
\begin{frontmatter}

\title{Security-Aware Planning and Control of Multi-Agent Systems with LTL Tasks} 

\author[First]{Georgios Mitsos} 
\author[Second]{Dimos V. Dimarogonas}
\author[First]{Siyuan Liu} 

\address[First]{\mbox{Department of Electrical Engineering, Eindhoven University of Technology}, Eindhoven, the Netherlands (e-mails: 
\{g.mitsos, s.liu5\}@tue.nl)}
\address[Second]{Division of Decision and Control Systems, KTH Royal Institute of Technology, Stockholm, Sweden (e-mail: dimos@kth.se)}

\begin{abstract}    
This paper presents a secure-by-construction planning and control framework for multi-agent systems subject to linear temporal logic (LTL) specifications. The framework protects sensitive information from a passive intruder with partial observations of the agents’ motion. Security in multi-agent coordination is captured by two notions that prevent the intruder from inferring whether a secret task has been executed and from identifying the agent responsible for its execution. 
The proposed framework incorporates the security constraints directly into the LTL synthesis procedure by constructing a secure finite transition system that removes all paths violating these constraints. Standard LTL synthesis is then applied to this secure abstraction to generate discrete plans, which are then refined into dynamically feasible continuous trajectories. This synthesis procedure provides formal guarantees that the resulting behavior of the multi-agent system satisfies both the global LTL specification and the security constraints. The effectiveness of the proposed framework is demonstrated through a two-drone case study.
\end{abstract}


\end{frontmatter}

\section{Introduction} 
Linear temporal logic (LTL) \citep{baier2008} can specify complex tasks beyond point-to-point navigation \citep{lavalle2006}. By abstracting agent dynamics to finite-state models, model-checking techniques enable synthesis of fully automated, correct-by-design discrete plans \citep{smith2011}, which are then refined into continuous trajectories via low-level controllers \citep{fainekos2009, kloetzer2008}.

In multi-agent systems, LTL enables high-level planning beyond classical cooperative tasks, e.g., collision avoidance, consensus, and formation. However, interdependent tasks and coordination requirements make LTL synthesis over the joint state space quickly intractable. To improve scalability, various decentralized methods have been proposed \citep{kloetzer2007, loizou2004, guo2014}.

Despite these advances, information security during task execution has received comparatively less attention. Autonomous cyber-physical systems (CPS), such as drone swarms and intelligent transportation systems, often operate in settings where task correctness and confidentiality are equally critical. This motivates secure-by-construction strategies that incorporate security constraints directly into the control synthesis process \citep{liu2022,yin2021}. Recent works have incorporated information-flow security in robotic path planning via the notion of \emph{opacity}; see, e.g.,  \citep{ma2021optimal,hadjicostis2018,wang2020hyperproperties,xie2021secure} for single-agent systems and \citep{yu2022} for multi-agent systems. However, these works do not address systems with continuous dynamics, where controllers must simultaneously ensure security and LTL satisfaction in the physical workspace.

In this paper, we propose a security-aware planning framework for multi-agent systems with continuous affine dynamics subject to LTL specifications. We consider the presence of an intruder modeled as a malicious observer with partial observations of the system. The main contributions of this work are summarized as follows:

\vspace{-1mm}
(i) We introduce two complementary information-flow security notions: Type-A Security, which prevents the intruder from inferring whether any agent has visited a secret region, and Type-B Security, which prevents identifying the agent responsible for the secret visit. These requirements can be enforced independently or jointly.

\vspace{-1.5mm}
(ii) Based on these notions, we construct secure Weighted Transition Systems (WTSs) and apply standard LTL synthesis to generate secure discrete plans, which are then refined into dynamically feasible continuous trajectories that satisfy both the LTL tasks and the security constraints.
Compared to existing LTL-based security-aware planning methods \citep{yang2020,yu2022,xie2021secure}, where systems are modeled as finite WTSs, our framework is, to the best of our knowledge, the first directly applicable to multi-agent systems with continuous dynamics, input constraints, and complex LTL specifications. Moreover, unlike \citep{zhong2025}, which focuses on opacity and safety, our framework addresses security-aware multi-agent planning under general LTL tasks.

\vspace{-1.5mm}
(iii) Under the proposed security notions, the resulting discrete planning algorithm achieves significantly lower complexity than that in \citep{yu2022} (cf. Section \ref{subsec:task_satisfaction}).

The remainder of the paper is organized as follows. Section~\ref{sec:preliminaries} introduces preliminaries and notation. Section~\ref{sec:security} formalizes the security notions for multi-agent systems, and states the planning problem.
Section~\ref{sec:secure_planning} describes the planning and control framework, ensuring both LTL satisfaction and security in continuous workspaces. Section~\ref{sec:case_study} demonstrates the proposed method through a two-drone case study, and Section~\ref{sec:conclusion} concludes the paper.

\vspace{-0.1cm}
\section{Preliminaries} \label{sec:preliminaries}
\vspace{-0.1cm}
\subsubsection{Notation.} Let $\mathbb{R}$ denote the set of real numbers and $\mathbb{N}$ the set of positive integers. We denote $\mathbb{R}_{\geq 0}$ and $\mathbb{R}_+$ the sets of nonnegative and  positive real numbers, respectively. For any $n,m \in \mathbb{N}$, $\mathbb{R}^n$ denotes the set of $n$-dimensional column vectors and $\mathbb{R}^{n \times m}$ the set of $n \times m$ real matrices. Vectors are denoted by bold lowercase letters and matrices by bold uppercase letters. A class-$\mathcal{K}$ function is a function $\alpha: \mathbb{R}_{\geq 0} \to \mathbb{R}_{\geq 0}$ that is strictly increasing and with $\alpha(0) = 0$. For any vector $\mathbf{x} \in \mathbb{R}^n$, $\|\mathbf{x}\|$ denotes the Euclidean norm, i.e., $\|\mathbf{x}\| = \sqrt{\mathbf{x}^\top \mathbf{x}}$.

\subsection{Polytopes}

A polytope is a bounded convex set in $\mathbb{R}^n$ and is full-dimensional if it has a non-empty interior in $\mathbb{R}^n$. Any full-dimensional polytope can be represented as the intersection of $d \geq n + 1$ closed half-spaces, i.e., $P = \big\{ \mathbf{x} \in \mathbb{R}^n \mid \mathbf{p}_i^\top \, \mathbf{x} + g_i \geq 0, \, i = 1, \dots, d \big\}$. 

A polytope $P \subset \mathbb{R}^n$ consists of lower-dimensional \emph{faces}, obtained by intersecting $P$ with its supporting hyperplanes. A $j$-face has dimension $j$ $(0 \leq j < n)$; facets are $(n-1)$-faces, edges are 1-faces, and vertices are 0-faces.

\subsection{System model} \label{subsec:system_model}

Consider a multi-agent system with $M$ agents, indexed by $\mathcal{A} = \{1, \dots, M\}$, operating in a workspace modeled as a full-dimensional polytope $P \subset \mathbb{R}^n$. Each Agent~$i \in \mathcal{A}$ follows the linear affine dynamics
\begin{equation} \label{eq:dynamics}
   \mathbf{\dot x}_i = \mathbf{A}_i \, \mathbf{x}_i + \mathbf{B}_i \, \mathbf{u}_i + \mathbf{b}_i,
\end{equation}
where $\mathbf{x}_i \in P$ is Agent~$i$'s state vector, $\mathbf{u}_i \in U_i \subset \mathbb{R}^m$ is the input vector, $\mathbf{A}_i \in \mathbb{R}^{n \times n}$ and $\mathbf{B}_i \in \mathbb{R}^{n \times m}$ are constant matrices, and $\mathbf{b}_i \in \mathbb{R}^n$ is a constant vector. The set $U_i$ is a convex polytope indicating bounded control inputs.

The continuous trajectory of Agent~$i$ is denoted by $\mathbf{x}_i(t): \mathbb{R}_{\ge 0} \to P$. The global trajectory of the system is the stacked vector of all agents’ trajectories, i.e., $\mathbf{x}(t) = \begin{bmatrix} \mathbf{x}_1(t)^\top & \dots & \mathbf{x}_M(t)^\top \end{bmatrix}^\top$.





\subsection{Weighted transition systems} \label{subsec:wts}

The goal of this paper is to synthesize a continuous trajectory for a multi-agent system as in \eqref{eq:dynamics} subject to security constraints (cf. Sec.~\ref{sec:security}) and LTL specifications (cf. Sec.~\ref{subsec:ltl}), defined over a partitioned abstraction of the agent's motion. The polytopic workspace is partitioned into non-overlapping sub-polytopes $q_1, \dots, q_p$ such that $P = \bigcup_{j=1}^p q_j$ and $q_j \cap q_k = \emptyset$, for $j \neq k$. 
Each Agent~$i$'s motion among regions is abstracted as transitions in a finite Weighted Transition System (WTS) $\mathcal{T}^i$, which may differ across agents. 

\begin{definition} \label{def:wts}
    A Weighted Transition System $\mathcal{T}^i$ is a tuple $\left( Q^i, Q_0^i, \longrightarrow_i, w_i, \mathcal{AP}_i, L_i \right)$, where:
    \begin{itemize}
        \item $Q^i = \{ q_1, \dots, q_p \}$ is the set of states, where each $q_j$ corresponds to Agent~$i$ being in region $q_j$
        \item $Q_0^i \subseteq Q^i$ is the set of initial states
        \item $\longrightarrow_i \subseteq Q^i \times Q^i$  is the transition relation which collects all the \emph{feasible transitions} for Agent~$i$
        \item $w_i : Q^i \times Q^i \to \mathbb{R_+}$ is the cost function that assigns to each transition $(q,q') \in \longrightarrow_i$ a positive weight $w_i(q, q')$

        \item $\mathcal{AP}_i = \{\pi_j \mid j \in \mathbb{N} \}$ is the set of atomic propositions used for representing some properties of interest
        \item $L_i : Q^i \to 2^{\mathcal{AP}_i}$ is the labeling function specifying the properties that hold true at each state
    \end{itemize}
\end{definition}

The \emph{dynamical feasibility} of the transitions in $\longrightarrow_i$ w.r.t. the continuous system \eqref{eq:dynamics} is characterized in Sec.~\ref{subsec:feasibility}. The abstracted motion of all $M$ agents, each executing one transition per time step, is represented by a Global Weighted Transition System (gWTS), defined as the synchronous product of the individual agents' WTSs.

\begin{definition}\label{def: gWTS} 
    A Global Weighted Transition System $\mathcal{T}_g$ is a tuple $\bigotimes_{i=1}^M \mathcal{T}^i = \left( Q_g, Q_{g,0}, \longrightarrow_g, w_g, \mathcal{AP}_g, L_g \right)$, where:
    \begin{itemize}
        \item $Q_g = Q^1 \times \dots \times Q^M$ is the set of global states
        
        \item $Q_{g,0} = Q_0^1 \times \dots \times Q_0^M \subseteq Q_g$ is the set of initial states
        
        \item $\longrightarrow_g \subseteq Q_g \times Q_g$ is the transition relation, where, for any $q, q' \in Q_g$, $(q, q') \in \longrightarrow_g$ if $(q^i,q^{i \, \prime}) \in \longrightarrow_i, \forall i \in \mathcal{A}$

        \item $w_g : Q_g \times Q_g \to \mathbb{R_+}$ is the cost function, where, for any $(q,q') \in \longrightarrow_g$, $w_g(q, q') = \sum_{i=1}^M w_i(q^i,q^{i \, \prime})$ 

        \item $\mathcal{AP}_g = 2^{\mathcal{AP}_1} \times \dots \times 2^{\mathcal{AP}_M}$

        \item $L_g : Q_g \to \mathcal{AP}_g$ is the labeling function, where, for any $q \in Q_g$, $L_g(q) = \big(L_1 \big(q^1 \big), \dots, L_M \big(q^M \big) \big)$
    \end{itemize}
\end{definition}

Each global state $q = \big(q^1, \dots, q^M \big) \in Q_g$ collects the local states of all agents. A local path of Agent~$i \in \mathcal{A}$ is a finite or infinite sequence of states $\tau_i = \tau_i(1) \tau_i(2) \dots$ in $\mathcal{T}^i$, where
\[\tau_i(1) \in Q^i_0 \ \text{and} \ (\tau_i(j),\tau_i(j+1)) \in \longrightarrow_i, \ \text{for all} \ j \geq 1.
\]
The trace of a local path is the sequence of labels along it, i.e., $trace(\tau_i) = L_i(\tau_i) = L_i(\tau_i(1)) \, L_i(\tau_i(2)) \dots$.

A global path $\tau = \tau(1) \tau(2) \dots$ in $\mathcal{T}_g$ satisfies
\[\tau(1) \in Q_{g,0} \ \text{and} \ (\tau(j),\tau(j+1)) \in \longrightarrow_g, \ \text{for all} \ j \geq 1,
\]
and has $trace(\tau) = L_g(\tau) = L_g(\tau(1)) \, L_g(\tau(2)) \dots$.

The sets of all finite and infinite global paths in $\mathcal{T}_g$ are denoted $Path^*(\mathcal{T}_g)$ and $Path^\omega(\mathcal{T}_g)$, respectively, and the set of all paths is $Path(\mathcal{T}_g) = Path^\omega(\mathcal{T}) \ \cup Path^*(\mathcal{T})$. 

A special class of infinite paths is the prefix-suffix form
\begin{equation} \label{eq:prefix-suffix}
    \tau = \tau(1) \dots \tau(k) [\tau(k+1) \dots \tau(N)]^\omega,
\end{equation}
where the prefix $\tau(1) \dots \tau(k)$ is executed once, and the suffix $\tau(k+1) \dots \tau(N)$ is a cycle that repeats indefinitely. It must hold that $(\tau(N), \tau(k+1)) \in \longrightarrow_g$. In the sequel, we will also write a path as $\tau = \tau(1) \rightarrow_g \tau(2) \rightarrow_g \dots$.

\subsection{Control barrier functions}\label{subsec: cbf}

Consider an Agent~$i$ as in \eqref{eq:dynamics} and 
a set $\mathcal{C} \subseteq P \subset \mathbb{R}^n$ which is defined as the superlevel set of a continuously differentiable function $h: \mathbb{R}^n \to \mathbb{R}$, i.e., 
\begin{equation} \label{eq:safe_set}
\mathcal{C} = \{\mathbf{x} \in P \ | \ h(\mathbf{x}) \geq 0 \}.
\end{equation}

\begin{definition}\label{def:cbf}
Consider a set defined as in \eqref{eq:safe_set}. Function $h(\mathbf{x})$ is a control barrier function (CBF) if there exists a locally Lipschitz class-$\mathcal{K}$ function $\alpha$ s.t. for each $\mathbf{x}_i \in P$: 
    \begin{equation} \label{eq:cbf}
   \! \sup_{\mathbf{u}_i \in U_i} \Bigg[ \Bigg(\frac{\partial h}{\partial \mathbf{x}_i}(\mathbf{x}_i) \Bigg)^\top \! \! (\mathbf{A}_i \, \mathbf{x}_i + \mathbf{B}_i \, \mathbf{u}_i + \mathbf{b}_i)\Bigg] \! \geq \! -\alpha(h(\mathbf{x}_i)).
\end{equation}
\end{definition}
The CBF definition \citep{ames2019} will be used later in Sec. \ref{subsec:feasibility} to guarantee \emph{forward invariance} properties of workspace partitions.  

\subsection{Linear temporal logic} \label{subsec:ltl}

LTL formulas are defined over a set of atomic propositions $\mathcal{AP}$ using Boolean and temporal operators. The syntax of LTL is given recursively \citep{baier2008}:
\[
    \phi \Coloneqq p \mid \lnot \, \phi \mid \phi_1 \land \phi_2 \mid \mathcal{X}\phi \mid \phi_1 \, \mathcal{U} \, \phi_2,
\]
where $p \in \mathcal{AP}$ and $\phi, \, \phi_1, \, \phi_2$ are LTL formulas; $\mathcal{X}$ and $\mathcal{U}$ denote the \emph{next} and \emph{until} operators, respectively. The Boolean disjunction ($\lor$) is derived as $\phi_1 \lor \phi_2 {:=} \lnot (\lnot \phi_1 \land \lnot \phi_2)$, the Boolean $true$ as $p \lor \lnot p$, while the temporal operators \emph{eventually} ($\lozenge$) and \emph{always} ($\square$) are defined as $\lozenge \phi  {:=} true \ \mathcal{U} \ \phi$, and $\square \, \phi {:=} \lnot  \lozenge \, \lnot \phi$.

LTL formulas are interpreted over infinite traces. A trace $r \in \left(2^\mathcal{AP} \right)^\omega$ satisfies an LTL formula $\phi$, written as $r \models \phi$, if $\phi$ holds over $r$ according to the LTL semantics \citep{baier2008}, which are omitted here for brevity.

\begin{definition} 
    A Nondeterministic B\"{u}chi Automaton (NBA) is a tuple $\mathcal{B}_\phi = (S, S_0, F, \Sigma, \delta)$, where $S$ is a finite set of states; $S_0 \subseteq S$ is the set of initial states; $F \subseteq S$ is the set of accepting states; $\Sigma = 2^\mathcal{AP}$ is the input alphabet; and $\delta: S \times \Sigma \to 2^S$ is the transition relation.
\end{definition}

It has been shown in \citep{wolper1983} that, for any LTL formula $\phi$ over $\mathcal{AP}$, there exists a NBA $\mathcal{B}_\phi$ with input alphabet $\Sigma = 2^\mathcal{AP}$ such that $\mathcal{B}_\phi$ accepts exactly the infinite traces that satisfy $\phi$. 

\section{Security Constraints in Multi-Agent Systems and Problem Formulation} \label{sec:security}

In this section, we introduce two new security properties for multi-agent systems, capturing whether a malicious intruder modeled as an outside observer can infer certain secret information. The problem formulation follows.

\subsection{Security constraints for multi-agent systems}
Consider a multi-agent systems whose motion is described by a gWTS as in Def.~\ref{def: gWTS}. We assume the passive intruder has knowledge of the individual agent models, i.e., $\mathcal{T}^i = \left( Q^i, Q_0^i, \longrightarrow_i, w_i, \mathcal{AP}_i, L_i \right)$, for all $i \in \mathcal{A}$, and it can observe each Agent~$i$ through an output function $H_i: Q^i \to Y_i$, where $Y_i$ is the output set of possible observation symbols.

To formulate the security requirement, we aim to prevent the intruder from inferring secret information based on available observations; specifically, whether sensitive behaviors occur and, when applicable, which agents execute them. These sensitive behaviors are modeled as visits to designated secret states
$Q_S^i \subseteq Q^i$ for each Agent~$i$.

To incorporate these security-related features, the standard WTS in Def.~\ref{def:wts} is extended as
\[
    \mathcal{T}^i = \left( Q^i, Q_0^i, \longrightarrow_i, w_i, \mathcal{AP}_i, L_i, Y_i, H_i, Q_S^i \right).
\]
The corresponding gWTS is rewritten as
\begin{equation} \label{eq:gwts}
    \begin{aligned}
        \mathcal{T}_g = \big( Q_g, Q_{g,0}, \longrightarrow_g, w_g, \mathcal{AP}_g, L_g, Y_g, H_g \big),
    \end{aligned}
\end{equation}
where $Y_g = Y_1 \times \dots \times Y_M$ and $H_g: Q_g \to Y_g$ with $H_g(q) = \big(H_1 \big(q^1 \big), \dots, H_M \big(q^M \big) \big)$. An output sequence is denoted as $y_g = y_g(1) \, y_g(2) \dots$, or equivalently, $y_g = y_g(1) \rightarrow_g y_g(2) \rightarrow_g \dots$.

The two types of security notions are formalized as follows.

\begin{definition} \label{def:type_a}
    Consider the gWTS in \eqref{eq:gwts}. A global path $\tau \in Path(\mathcal{T}_g)$ is \emph{Type-A secure} if there exists a global path $\tau' \in Path(\mathcal{T}_g)$ such that:  
    \begin{itemize}
        \item $H_g(\tau) = H_g(\tau')$, and  
        \item for every $i \in \mathcal{A}$ and $j \geq 1$: $\tau_i(j) \in Q_S^i \Rightarrow \tau_i'(j) \notin Q_S^i$.  
    \end{itemize}
\end{definition}

Intuitively, this security property states that whenever an agent visits a secret state, there exists an alternative behavior of the system that is indistinguishable to the intruder, and in which the agent is not in a secret state at that time. This provides plausible deniability: the intruder cannot be certain if any specific agent has actually visited a secret state. In the sequel, we refer to $\tau$ as the ``real path" and to $\tau'$ as its ``copy path".

\begin{remark}
    For Type-A security to be meaningful, each Agent~$i \in \mathcal{A}$ must have at least one non-secret state, i.e. $Q_i \setminus Q_S^i \neq \emptyset$, and a non-secret initial state, i.e., $Q_0^i \setminus Q_S^i \neq \emptyset$. Otherwise, the property is violated trivially.
\end{remark}

Type-A security treats each agent independently; that is, there is no collaboration between agents to guarantee security. For each agent, satisfaction of Type-A security for all paths implies satisfaction of classical state-based opacity notions, including initial-state, current-state, and infinite-step opacity \citep{saboori2007}.
Note that Type-A security does not prevent the intruder from identifying an agent as the only candidate for a secret task, e.g., if the intruder knows the secret task location and only one agent produces the corresponding observation. Type-B security addresses such cases.

\begin{definition} \label{def:type_b}
    Consider the gWTS in \eqref{eq:gwts}. A global path $\tau \in Path(\mathcal{T}_g)$ is \emph{Type-B secure} if, for every Agent~$i \in \mathcal{A}$ and every time step $j \geq 1$: 
    \[
        \tau_i(j) \in Q_S^i \, \Rightarrow \, \exists k \neq i \ \text{s.t.} \ H_k(\tau_k(j)) = H_i(\tau_i(j)).
    \] 
\end{definition}
\vspace{-0.2cm}
Intuitively, Type-B security ensures that whenever an Agent~$i$ visits a secret state, at least one other agent $k$ produces the same observation at that time, preventing the intruder from identifying which agent is responsible for the secret visit. A path that satisfies both Type-A and Type-B security is called \emph{Type-A/B secure}.

Similar security notions for multi-agent systems have been introduced in \citep{yu2022}, which also capture whether an intruder can infer a secret state has been visited (Type-I) and by which agent (Type-II). 
However, unlike Type I, which requires the alternative (copy) path to never visit any secret state, our Type-A security notion imposes this restriction only at the time steps when an agent is in a secret state. Similarly, unlike Type II, which requires that some other agent visit a secret state, our Type-B leverages the intruder’s limited knowledge of the system's evolution and only requires that another agent produces the same observation at those steps. 
These relaxations reduce the size of the problem while preserving similar security guarantees, making our approach considerably more scalable than the one in \citep{yu2022}; a detailed complexity analysis is provided in Section~\ref{subsec:secure_planning_alg}.

We use the following simple example to illustrate the proposed security notions. 
\begin{example}    
Consider a two-agent system operating in a workspace shown in Fig.~\ref{fig:motivating_example}, where the blue circle and the green square indicate the location of Agents~1 and 2, respectively. Secret states are marked in red. At each time step, an agent may only move between adjacent states. 
We present one path satisfying both security notions, and another that violates them. 
Consider a global (finite) path $\tau = (D,E) \rightarrow_g (E,B)$ of the system as shown in Fig.~\ref{fig:motivating_example}$(i)$, with output sequence $H_g(\tau) = (y_3, y_2) \to_g (y_2, y_2) $.  Note that at the second time step, Agent~2 visits secret state $B$, i.e.,  $\tau_2(2) \in Q_S^2$. 
According to Def.~\ref{def:type_a}, the global path $\tau$ is Type-A secure since there exists another path $\tau' = (C,E) \rightarrow_g (B,E)$ with the same output sequence $H_g(\tau') = (y_3, y_2)\to_g(y_2, y_2)$, where $\tau_2'(2) \notin Q_S^2$. Thus, from the intruder's point of view, these two paths are indistinguishable.
Moreover, path $\tau$ is Type-B secure according to Def.~\ref{def:type_b}, since: for Agent~$2$ with $\tau_2(2) \in Q_S^2$, there exists Agent~$1$ s.t. $H_1(\tau_1(2)) = H_2(\tau_2(2)) = y_2$, which prevents the intruder from identifying which agent performs the secret task. Copy paths/agents are indicated with dotted outlines in Fig.~\ref{fig:motivating_example}.

Consider a global path $\tau = (A,B) \to_g (F,A)$, as shown in Fig.~\ref{fig:motivating_example}$(ii)$, with output sequence $H_g(\tau) = (y_1, y_2) \to_g (y_1, y_1)$. Both agents visit secret states at all time steps, i.e., $\tau_1(1), \tau_1(2) \in Q_S^1$, and $\tau_2(1), \tau_2(2) \in Q_S^2$. The path $\tau$ is not Type-A secure, since there does not exist a ``copy path" that satisfies the conditions in Def.~\ref{def:type_a}. In particular, for Agent~$1$, there is no path $\tau'_1$ with $\tau'_1 (j) \notin Q_S^i$, $j=1,2$, such that $H_g(\tau) = H_g(\tau')$. Likewise, for Agent~2, there is no $\tau'_2(2) \notin Q_S^2$ such that $H_2(\tau_2(2)) = H_2(\tau_2'(2))$.
Moreover, $\tau$ violates Type-B security at the first step, as both agents are in secret states and their observations differ, i.e., $H_1(\tau_1(1)) \neq H_2(\tau_2(1))$.
\end{example}

\begin{figure}[tb]
    \centering
    \begin{minipage}{0.48\linewidth}
        \centering
        \includegraphics[width=\linewidth]{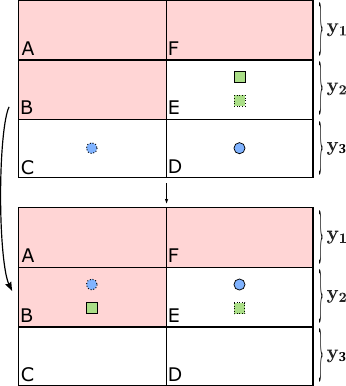}
        \centerline{$(i)$ Type-A/B secure path}
    \end{minipage}
    \hfill
    \begin{minipage}{0.48\linewidth}
        \centering
        \includegraphics[width=\linewidth]{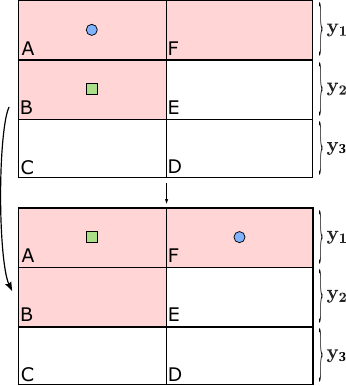}
        \centerline{$(ii)$ Non-secure path}
    \end{minipage}
    \caption{Two agents operate in a workspace. States in the same row generate the same observation $y_i$, $i = 1,2,3$. }
    \label{fig:motivating_example}
\end{figure}

\subsection{Problem formulation}

In the previous subsections, we defined the desired LTL formulas and security properties over discrete paths and their traces. Note that the ultimate goal is to design a planning and control scheme for agents evolving in continuous time and space. Therefore, given a discrete path, we require that every discrete transition corresponds to a realizable continuous motion between designated locations of the workspace.

\begin{assumption} \label{asmp:centroids}
For each discrete transition, Agent~$i$ moves between representative points of the corresponding regions in $\mathcal{T}^i$. Without loss of generality, we select these points as the region centroids. Thus, a global transition $(q, q')$ starts with all Agents~$i$ at the centroid of region $q^i$, and ends with all agents at the centroid of region $q^{i \, \prime}$.
\end{assumption}

Assumption~\ref{asmp:centroids} may exclude some feasible trajectories but allows us to guarantee the existence of a continuous controller for every discrete transition. The next definition formalizes how a global path is induced as the sequence of visited sub-polytopes.

\begin{definition} \label{def:discrete_continuous}
    Let $\mathbf{x}_i(t)$ be the trajectory of Agent~$i$, and define $f_i(t) = q_j$  when $\mathbf{x}_i(t) \in q_j$. The switching times of Agent~$i$ are defined recursively as
    \[
        t_{i,1} = 0, \quad t_{i,k+1} = \inf \Big\{ t > t_{i,k} \, | \, f_i(t) \neq f(t_{i,k}) \Big\}.
    \]
    This yields $\tau_i = \tau_i(1) \, \tau_i(2) \dots$, with $\tau_i(j) = f_i(t_{i,j})$. Let $(s_\ell)_{\ell \in \mathbb{N}}$ be the increasing ordering of all switching times across all agents. At each $s_\ell$, the global state $\tau(\ell) = (\tau_1(\ell), \dots, \tau_M(\ell))$ is defined, and all agents $i \in \mathcal{A}$ transition from $\tau_i(\ell - 1)$ to $\tau_i(\ell)$ simultaneously. If the system eventually remains in a global state $\tau(k)$, then $\tau(k)$ repeats indefinitely, generating an infinite global path.
\end{definition}

\begin{remark} \label{rem:ltl_security}
    By Def.~\ref{def:discrete_continuous}, each global trajectory $\mathbf{x}$ induces a unique global path $\tau$, representing the sequence of visited sub-polytopes. Since labels, observations, and secret states are constant within each sub-polytope, $\mathbf{x}(t)$ and $\tau$ produce identical traces. Hence, we write $trace(\mathbf{x}) \models \phi$ to indicate that a continuous trajectory satisfies an LTL formula $\phi$ with a slight abuse of notations.  We then have that $trace(\mathbf{x}) \models \phi$    
    if and only if $trace(\tau) \models \phi$, and, Type-A/B security of $\tau$ implies the same for $\mathbf{x}$.
\end{remark}

This paper first aims to compute a discrete plan that satisfies the global LTL specification and security requirements, and then synthesize a low-level (piecewise) continuous controller that implements the motion of the multi-agent system corresponding to each discrete transition.

\begin{problem} \label{pr:secure_ltl_planning}
    \textbf{Security-aware multi-agent LTL planning and control} \\
    Given an input-constrained  multi-agent system as in \eqref{eq:dynamics} operating in a partitioned workspace $P$ and a global task specified by an LTL formula $\phi$, synthesize a controller such that the generated global trajectory $\mathbf{x}: \mathbb{R}_{\geq 0} \to P$ satisfies:
    \begin{enumerate}
        \item $trace(\mathbf{x}) \models \phi$, and
        \item $\mathbf{x}$ is Type-A and/or Type-B secure.
    \end{enumerate}
\end{problem}

\section{Secure-by-Construction Planning and Control} \label{sec:secure_planning}

In this section, we present the planning and control algorithm for solving Problem~\ref{pr:secure_ltl_planning}. First, secure discrete abstractions of the multi-agent system are constructed, retaining only global paths that satisfy the security constraints. Next, the feasibility of each transition w.r.t. system dynamics and input constraints is verified. LTL synthesis is then employed to compute a prefix-suffix path that satisfies the global task. Finally, a continuous trajectory is generated from the discrete path, guaranteeing satisfaction of both the global LTL task and the security constraints.

\subsection{Secure global transition systems} \label{subsec:secure_gwts}
 
\subsubsection{Type-A security.}

Type-A security ensures that the intruder can never infer that any agent is in a secret state. To enforce it, we construct two abstractions: the Twin-Global Weighted Transition System (Twin-gWTS), which captures pairs of indistinguishable global states, and the Secure Twin-gWTS, which removes all pairs where any agent’s real and copy local states are both secret. 

The concept of a Twin-Weighted Transition System was introduced in \citep{yang2020} and is extended here to multi-agent systems. 

\begin{definition} \label{def:global_twin_wts}
    Let $\mathcal{T}_g$ be a gWTS in \eqref{eq:gwts}. The corresponding Twin-gWTS is 
\[    \mathcal{V} = \big( X, X_0, \longrightarrow_V, w_V, \mathcal{AP}_V, L_V \big),
\]  
    where:
    \begin{itemize}
        \item $X \subseteq Q_g \times Q_g$ is the set of states, with $(q_1, q_2) \in X$ if $q_1, q_2 \in Q_g$ and $H_g(q_1) = H_g(q_2)$
        \item $X_0 = X \cap (Q_{g,0} \times Q_{g,0})$ is the set of initial states
        \item $\longrightarrow_V \subseteq X \times X$ is the transition relation, where, for any $(q_1,q_2), (q_1', q_2') \in X$, $((q_1,q_2),(q_1',q_2')) \in \longrightarrow_V$ if: $(q_1,q_1') \in \longrightarrow_g$, $(q_2,q_2') \in \longrightarrow_g$, and $H_g(q_1^\prime) = H_g(q_2^\prime)$
        \item $w_V : X \times X \to \mathbb{R_+}$ is the cost function, where, for any $((q_1, q_2), (q_1', q_2')) \in \longrightarrow_V$, $w_V((q_1, q_2), (q_1', q_2')) = w_g(q_1, q_1')$
        \item $\mathcal{AP}_V = \mathcal{AP}_g$
        \item $L_V : X \to \mathcal{AP}_V$ is the labeling function, where, for any $(q_1,q_2) \in X$, $L_V((q_1,q_2)) = L_g(q_1)$
    \end{itemize}
\end{definition}

Each twin state $q = (q_1, q_2) = \left( \left(q_1^1, \dots, q_1^M), (q_2^1, \dots, q_2^M \right) \right)$ in the Twin-gWTS represents a pair of indistinguishable global states, i.e., $H_g(q_1) = H_g(q_2)$, where $q_1$ is the ``real state" and $q_2$ the ``copy state". Intuitively, the system is in $q_1$, but the intruder cannot distinguish it from $q_2$.  

\begin{definition}\label{def:global_twin_swts}
    Let $\mathcal{V}$ be a Twin-gWTS in Def.~\ref{def:global_twin_wts}. The corresponding Secure Twin-gWTS is 
    \[
    \begin{aligned}
        \mathcal{V}_s = \big( X_s, X_{s,0}, \longrightarrow_s, w_s, \mathcal{AP}_s, L_s \big),
    \end{aligned}
    \]
    and is obtained by removing all twin states $(q_1, q_2)$, in which there is an Agent~$i \in \mathcal{A}$ such that $q_1^i, q_2^i \in Q_S^i$. All other components are inherited from $\mathcal{V}$ in Def.~\ref{def:global_twin_wts}.
\end{definition}

For a twin state $x_s = (q_1, q_2) \in X_s$, denote $\Pi(x_s) = q_1$. For a twin path $\sigma = \sigma(1) \sigma(2) \dots \in Path(\mathcal{V}_s)$, the projected path in $Path(\mathcal{T}_g)$ is $\Pi(\sigma) = \Pi(\sigma(1)) \Pi(\sigma(2)) \dots$.

\subsubsection{Type-B security.}

Type-B security addresses scenarios in which the intruder can identify an agent as the only candidate for performing a secret task. To enforce it, we construct the Type-B gWTS as follows. 

\begin{definition} \label{def:gwts_type_b}
Let $\mathcal{T}_g$ be a gWTS in \eqref{eq:gwts}. The corresponding Type-B gWTS $\mathcal{T}_g^B$ is
\[
    \mathcal{T}_g^B = \big( Q_g^B, Q_{g,0}^B, \longrightarrow_g^B, w_g^B, \mathcal{AP}_g^B, L_g^B \big),
\]
and is obtained by restricting the global states to
\[ 
\!Q_g^B \!=\!  \Big\{ \!(q^1\!,\!\dots\!,q^M) \!\in\! Q_g \!\ |  q^i \!\in\! Q_S^i \!\Rightarrow \!\exists k \neq i\!:\! H_i(q^i) \!=\! H_k(q^k) \Big\}.
\]
All other components are inherited from $\mathcal{T}_g$ in \eqref{eq:gwts}.
\end{definition}

\subsubsection{Type-A/B security.}

Type-A and Type-B security can be enforced simultaneously by first constructing the Type-B gWTS $\mathcal{T}_g^B$, and then building the Twin-gWTS $\mathcal{V}$ and Secure Twin-gWTS $\mathcal{V}_s$ on top of it. The following propositions are provided without proofs, as they hold trivially by the construction of the secure gWTSs.

\begin{prop} \label{prop:security_correctness} By construction, every path $\sigma \in Path(\mathcal{V}_s)$ satisfies both Type-A and Type-B security.
\end{prop}

\begin{prop} \label{prop:security_completeness} For every path $\tau \in Path(\mathcal{T}_g)$ that satisfies both Type-A and Type-B security, there exists at least one path $\sigma \in Path(\mathcal{V}_s)$ such that $\Pi(\sigma)=\tau$.  
\end{prop}
Proposition \ref{prop:security_completeness} shows that the constructed $\mathcal{V}_s$ captures all paths that satisfy both Type-A and Type-B security, yielding completeness guarantees.

\subsection{Transition feasibility} \label{subsec:feasibility}

A discrete transition $(q, q') \in \longrightarrow_s$, with $q \neq q'$, is said to be \emph{dynamically feasible w.r.t. system \eqref{eq:dynamics}} if there exists a controller $\mathbf{u}_i $, for each real and copy Agent~$i \in \mathcal{A}$, such that the trajectory $\mathbf{x}_i(t): [0, t_f] \to P$ drives the agent from the centroid $\mathbf{\xi}_i$ of the current sub-polytope $q^i$ to the centroid $\mathbf{\xi}_i'$ of the next sub-polytope $q^{i \, \prime}$, satisfying: (i)~the agent's dynamics in \eqref{eq:dynamics}; (ii)~input constraint $\mathbf{u}_i(t) \in U_i$; (iii)~initial and final centroid positions; (iv)~if $q^i \neq q^{i \, \prime}$, Agent~$i$ crosses the exit facet $F_{c,i} = q^i \cap q^{i \, \prime}$ at time $t = t_c$ as in Def.~\ref{def:discrete_continuous}; and (v)~Agent~$i$ remains within $q^i$ for $t \in T_1 = [0, t_c]$, and within $q^{i \, \prime}$ for $t \in T_2 = \left[ t_c, t_f \right]$. Real and copy Agent~$i \in \mathcal{A}$ share the same dynamics and input bounds.

\begin{lem} \label{lem:invariance_cbf}
    Let $q^i \subset \mathbb{R}^n$ be a convex polytope with facets $F_j \in q^i$, each represented by $h_j(\mathbf{x}) = \mathbf{p}_j^\top \mathbf{x} + g_j = 0$. Suppose there exists a locally Lipschitz continuous controller $\mathbf{u}_i: T_1 \to U_i$ s.t. for all $F_j \in q^i$ and $t \in T_1$, 
    \begin{equation}\label{cbf:facet}
    \mathbf{p}_j^\top (\mathbf{A}_i \, \mathbf{x}_i(t) + \mathbf{B}_i \, \mathbf{u}_i(t) + \mathbf{b}_i) \geq -\alpha(h_j(\mathbf{x}_i(t))),
    \end{equation}
where 
    $\alpha(r) = \gamma \, r$ with $\gamma > 0$. Then, $q^i$ is \emph{forward invariant} for $t \in T_1$. Likewise, if the condition holds for all $F_j \in q^{i \, \prime}$ and $t \in T_2$, then $q^{i \, \prime}$ is forward invariant for $t \in T_2$.
\end{lem}
\begin{pf} 
Note that the functions $h_j$ satisfying \eqref{cbf:facet} are CBF as in Def.~\ref{def:cbf}. It has been shown that any locally Lipschitz continuous controller 
$\mathbf{u}_i(t): [0, t_c] \to U_i$ that satisfies the CBF constraints \eqref{cbf:facet} on all facets $F_j$ of $q^i$ renders  $q^i$ 
\emph{forward invariant} \citep{ames2019}, i.e., $  \mathbf{x}_i(0) \in q^i \Rightarrow \mathbf{x}_i(t) \in q^i, \ \forall t \in T_1 = [0, t_c]$. The same reasoning applies to $q^{i \, \prime}$.  \hfill    $\blacksquare$
\end{pf}

Lemma~\ref{lem:invariance_cbf} ensures that by applying $\mathbf{u}_i$ satisfying \eqref{cbf:facet}, the closed-loop trajectory $\mathbf{x}_i$ remains within $q^i$ for all $t \in T_1$, and  $q^{i \, \prime}$ for all $t \in T_2$, and thus, feasibility condition (v) is satisfied. Consequently, if $q^i \neq q^{i \, \prime}$, we have $\mathbf{x}_i(t_c) \in F_{c,i} = q^i \cap q^{i \, \prime}$, i.e., $h_{c,i}(\mathbf{x}_i(t_c)) = 0$, satisfying condition (iv). 

To determine if a transition $(q, q') \in \longrightarrow_s$ is feasible, we check if there exist controllers, one for each real and each copy agent, that drive each agent between the corresponding sub-polytope centroids while satisfying conditions (i)-(v). This is done by solving a centralized optimization program for each $(q, q')$.

\begin{problem}
\label{pr:trajectory_optimization}
\textbf{Multi-agent trajectory optimization.} \\
For each real and copy Agent~$i$, compute a trajectory $\mathbf{x}_i$ and a locally Lipschitz continuous controller $\mathbf{u}_i$ that realizes $(q, q')$ while satisfying feasibility conditions (i)–(v), by solving the following centralized quadratic program (QP): 
    \begin{align*}
        \min_{\mathbf{x}_i, \mathbf{u}_i,\, i \in \mathcal{A}} \quad & \int_0^{t_f} \sum_{i \in \mathcal{A}} ||\mathbf{u}_i(t)||^2 dt  \\
        \text{s.t.} \quad 
            & \dot{\mathbf{x}}_i(t) = \mathbf{A}_i \mathbf{x}_i(t) + \mathbf{B}_i \mathbf{u}_i(t) + \mathbf{b}_i, \ \forall i \in \mathcal{A}  \\
            & \mathbf{u}_i(t) \in \mathcal{U}_i, \ \forall i \in \mathcal{A} \\
            & \mathbf{x}_i(0) = \mathbf{\xi}_i, \quad \mathbf{x}_i(t_f) = \mathbf{\xi}_i', \ \forall i \in \mathcal{A} \\
            & h_{c,i}(\mathbf{x}_i(t_c)) = \mathbf{p}_{c,i}^\top \mathbf{x}_i(t_c) + g_{c,i} = 0, \ \forall i \in \mathcal{A} \\
            & \mathbf{p}_j^\top (\mathbf{A}_i \, \mathbf{x}_i(t) + \mathbf{B}_i \, \mathbf{u}_i(t) + \mathbf{b}_i) \geq -\gamma \, h_j(\mathbf{x}_i(t))), \\ & \qquad \qquad \qquad \qquad \quad \forall i \in \mathcal{A}, \, F_j \in q^i, \, t \in [0, t_c]\\
            & \mathbf{p}_j^\top (\mathbf{A}_i \, \mathbf{x}_i(t) + \mathbf{B}_i \, \mathbf{u}_i(t) + \mathbf{b}_i) \geq -\gamma \, h_j(\mathbf{x}_i(t))), \\ & \qquad \qquad \qquad \quad \quad \ \forall i \in \mathcal{A}, \, F_j \in q^{i \, \prime}, \, t \in [t_c,t_f]
    \end{align*}
\end{problem}
\begin{remark}
The above QP is solved numerically in a sampled-data fashion over a fixed time horizon $t_f$, i.e., the control input is constant over each time step, the dynamics are approximated via a forward-Euler scheme, and the objective integral is discretized as a sum. Despite employing a centralized planning approach, the optimization program is convex and can be solved efficiently. If the QP admits a solution, then the discrete transition $(q, q') \in \longrightarrow_s$ is feasible: conditions (i)-(iii) are enforced by the first three constraints of the QP, while conditions (iv)-(v) follow from Lemma \ref{lem:invariance_cbf} and are enforced by the last three constraints.
Otherwise, the discrete transition $(q, q')$ is considered dynamically infeasible and is discarded from the WTS.
\end{remark}

\begin{remark}
Since sub-polytopes are defined by non-strict inequalities, trajectories remaining on facets over finite time intervals may create ambiguities in labels, observations, and secret status. To avoid this, a small slack variable $\epsilon>0$ is introduced in all CBF constraints:
    \[
      \mathbf{p}_j^\top (\mathbf{A}_i \, \mathbf{x}_i(t) + \mathbf{B}_i \, \mathbf{u}_i(t) + \mathbf{b}_i) \geq -\gamma \, (h_j(\mathbf{x}_i(t))) - \epsilon).
    \]
Subtracting $\epsilon$ from $h_j$ shifts the zero level set $h_j = 0$, i.e., the sub-polytope boundary, inward, ensuring trajectories remain strictly in the interior of each sub-polytope, except at facet-crossing instants $t=t_c$.
\end{remark}


\vspace{-0.1cm}
\subsection{LTL task satisfaction} \label{subsec:task_satisfaction}

In order to incorporate a given LTL specification $\phi$, we synchronize the resulting secure Twin gWTS $\mathcal{V}_s$ with the NBA $\mathcal{B}_\phi$ that accepts $\phi$
to form a Product B\"{u}chi Automaton (PBA).

\begin{definition} 
    A Product B\"{u}chi Automaton is
 $\mathcal{A}_P = \mathcal{V}_s \otimes \mathcal{B}_\phi  = (S_P, S_{P,0}, F_P, \delta_P, w_P)$,
    where:
    \begin{itemize}
        \item $S_P = X_s \times S$ is the set of product states
        \item $S_{P,0} = X_{s,0} \times S_0$ is the set of initial states
        \item $F_P = X_s \times F$ is the set of accepting states
        \item $\delta_P \subseteq S_P \times S_P$ is the transition relation, where, $((q,s), (q',s')) \in \delta_P$ iff $(q,q') \in \longrightarrow_s$ and $\exists \sigma \in L_s(q')$ such that $s' \in \delta(s,\sigma)$
        \item $w_P: \delta_P \to \mathbb{R}_+$ is the cost function, where, for any $((q,s), (q',s')) \in \delta_p$, $w_P((q,s), (q',s')) = w_s(q,q')$
    \end{itemize}
\end{definition}

The PBA further restricts the dynamics of $\mathcal{V}_s$ to satisfy $\phi$. 
A path $\tau$ in $\mathcal{V}_s$ satisfies $\phi$ if and only if it corresponds to an accepting path in the PBA, i.e., a path that starts from an initial state in $S_{P,0}$ and visits accepting states in $F_P$ infinitely often.  
Furthermore, if the PBA admits at least one accepting state, it also admits at least one in prefix-suffix form \eqref{eq:prefix-suffix} \citep{baier2008}.
The cost of $\tau$ is defined as $J(\tau) = \beta \, J_{\text{prefix}}(\tau) + (1 - \beta) \, J_{\text{suffix}}(\tau), \ \beta \in [0,1]$, with $J_{\text{prefix}}(\tau) = \sum_{i=1}^k w(\tau(i), \tau(i+1)),$ and $J_{\text{suffix}}(\tau) =\sum_{i=k+1}^{N-1} w(\tau(i), \tau(i+1)) \ + \ w(\tau(N), \tau(k+1))$.  

Finding an optimal discrete plan reduces to identifying a prefix-suffix path $\tau$ in the product system $\mathcal{A}_P$ that minimizes $J(\tau)$.
In the next subsection, we provide our planning and control algorithm to find an optimal trajectory that satisfies both an LTL formula and security constraints. 

\vspace{-0.2cm}
\subsection{Security-aware multi-agent planning and control} \label{subsec:secure_planning_alg}

\begin{algorithm}[b]
\caption{Secure Trajectory Planning and Control} \label{alg:secure_optimal_trajectory_planning}
\begin{algorithmic}[1]
    \Require \ LTL formula $\phi$, secure gWTS $\mathcal{W}$,  weight $\beta$
    \Ensure \ secure trajectory $\mathbf{x}^\star$, control $\mathbf{u}^\star$, cost $J^\star$ 

    \State Convert LTL formula $\phi$ to NBA $\mathcal{B}_\phi$
    
    \State Build PBA $\mathcal{A}_P$ from $\mathcal{W}$ and $\mathcal{B}_\phi$
    
    \State Initialize $\tau^\star \gets \emptyset$ and $J^\star \gets \infty$

    \For{each accepting state $q_\text{acc} \in F_P$}
    
        \State Compute shortest prefix path $\pi_\text{prefix}$ from any initial state $q_0 \in S_{P,0}$ to $q_\text{acc}$ using Dijkstra's algorithm
    
        \State Compute shortest suffix path $\pi_\text{suffix}$ starting and ending at $q_\text{acc}$ using Dijkstra's algorithm
    
        \State Compute total cost $J_\text{total} = \beta \cdot J_\text{prefix} + (1 - \beta) \cdot J_\text{suffix}$
    
        \If{$J_\text{total} < J^*$}
    
            \State $J^* \gets J_\text{total}$
        
            \State $\tau^* \gets \pi_\text{prefix} \oplus \pi_\text{suffix}$
        \EndIf
    \EndFor

    \State Reconstruct trajectory $\mathbf{x}^\star(t)$ and control $\mathbf{u}^\star(t)$ from $\tau^\star$
    
    \State \Return $\mathbf{x}^\star(t)$, $\mathbf{u}^\star(t)$, $J^\star$
\end{algorithmic}
\end{algorithm}

This section presents Alg.~\ref{alg:secure_optimal_trajectory_planning}, which computes a controller such that the resulting continuous multi-agent trajectory satisfies both the global LTL task and the security properties. 
Before executing the algorithm, the workspace is partitioned, and the motion of multi-agent system is abstracted as a gWTS $\mathcal{T}_g$ as in Sec.~\ref {subsec:wts}. 
Then, a secure gWTS is constructed as in Sec.~\ref{subsec:secure_gwts} to satisfy the security properties: the Secure Twin-gWTS $\mathcal{V}_s$ as in Def.~\ref{def:global_twin_swts} for Type-A security, and the Type-B gWTS $\mathcal{T}_g^B$ as in Def.~\ref{def:gwts_type_b} for Type-B. 
The transition feasibility in the resulting secure gWTS $\mathcal{W}$ (i.e., $\mathcal{V}_s$ or $\mathcal{T}_g^B$, depending on the security type) is checked via the optimization program in Sec.~\ref{subsec:feasibility}.

In Alg.~\ref{alg:secure_optimal_trajectory_planning}, line 1 converts the LTL formula $\phi$ to an NBA $\mathcal{B}_\phi$, and line 2 builds the PBA $\mathcal{A}_P$ from $\mathcal{W}$ and $\mathcal{B}\phi$. Line 3 initializes the optimal path and cost. Lines 4–10 iterate over all accepting states: for each $q_\text{acc}$, the shortest prefix path to the accepting state and the shortest suffix cycle are computed using Dijkstra’s algorithm, and their total cost is calculated (lines 5-7); if it improves the current minimum, the corresponding prefix-suffix path is stored as $\tau^\star$ (line 10). Finally, line 11 constructs the continuous trajectory $\mathbf{x}^\star(t)$ and control $\mathbf{u}^\star(t)$ from $\tau^\star$ by concatenating the continuous motion segments $\mathbf{x}^{q, q'}$ obtained from the optimization program in Sec.~\ref{subsec:feasibility}, one for each discrete transition $(q, q')$ in $\tau^\star$. Line 12 returns them along with the optimal cost $J^\star$. This ensures that the resulting trajectories satisfy both the LTL specification and the desired security properties.

\begin{thm}
Consider a multi-agent system as in \eqref{eq:dynamics} subject to an LTL formula $\phi$ and Type-A/B security constraints.
    Alg.~\ref{alg:secure_optimal_trajectory_planning} computes a continuous trajectory $\mathbf{x}^\star(t)$ that: (i)~satisfies $\phi$, (ii)~guarantees Type-A, and/or Type-B security, and (iii)~is dynamically feasible.
\end{thm}\vspace{-0.3cm}
\begin{pf}
    The optimal discrete plan $\tau^\star$ computed by Alg.~\ref{alg:secure_optimal_trajectory_planning} satisfies the LTL formula $\phi$, as it visits an accepting state in the PBA $\mathcal{A}_P$ infinitely often, and the desired security properties, as shown in Prop.~\ref{prop:security_correctness}. The continuous trajectory $\mathbf{x}^\star(t)$ is obtained by concatenating the trajectory segments $\mathbf{x}^{q,q'}(t)$ corresponding to each discrete transition $(q,q')$ in $\tau^\star$. By construction, each segment $\mathbf{x}^{q,q'}(t)$ respects the agents' dynamics and input constraints, hence $\mathbf{x}^\star(t)$ is dynamically feasible as discussed in Sec.~\ref{subsec:feasibility}. Moreover, Lemma~\ref{lem:invariance_cbf} guarantees 
    $f_i(t) = q_i$ for all $t \in [0, t_c]$, and, $f_i(t) = q_i'$ for all $t \in [t_c, t_f]$, where $f_i$ is the function defined in Def.~\ref{def:discrete_continuous}. Then, by Remark~\ref{rem:ltl_security}, the concatenated trajectory $\mathbf{x}^\star(t)$ satisfies $\phi$ and the same security properties as $\tau^\star$, which concludes the proof.
 \hfill    $\blacksquare$
\end{pf}
\vspace{-0.2cm}
\subsubsection{Complexity analysis.} \label{subsec:complexity} 
Consider a multi-agent system with $M$ agents, where each agent's abstracted motion is represented by a WTS $\mathcal{T}^i$ with at most $|Q|$ states. Let the NBA for the LTL formula have $|S|$ states. The gWTS $\mathcal{T}_g$ is the synchronous product of all $\mathcal{T}^i$, yielding at most $|Q|^M$ states. For Type-A/B security, the Type-B gWTS $\mathcal{T}_g^B$ is first constructed from $\mathcal{T}_g$, followed by the Twin-gWTS $\mathcal{V}$ and the Secure Twin-gWTS $\mathcal{V}_s$ as in Defs.~\ref{def:global_twin_wts}-\ref{def:global_twin_swts}, resulting in at most $|Q|^{2M}$ states for $\mathcal{V}_s$ and $|Q|^{2M}|S|$ states for the PBA. Thus, Alg.~\ref{alg:secure_optimal_trajectory_planning} performs a shortest-path search over a graph with at most $|Q|^{2M}|S|$  states. Note that Type-A security has the same complexity, while Type-B has at most $|Q|^M |S|$ states.
Compared to the approach in \citep{yu2022} which requires a graph search over $(2|Q|)^{2M^2+M} |S|$ states, our framework requires significantly lower computational complexity. In all cases, the dominant contribution to complexity comes from constructing the gWTS, while the contribution of the NBA is typically minor.

\vspace{-0.2cm}
\section{Case Study} \label{sec:case_study}
\vspace{-0.1cm}
\begin{figure}[t]
    \centering
    \begin{minipage}{0.95\linewidth}
        \includegraphics[width=\linewidth]{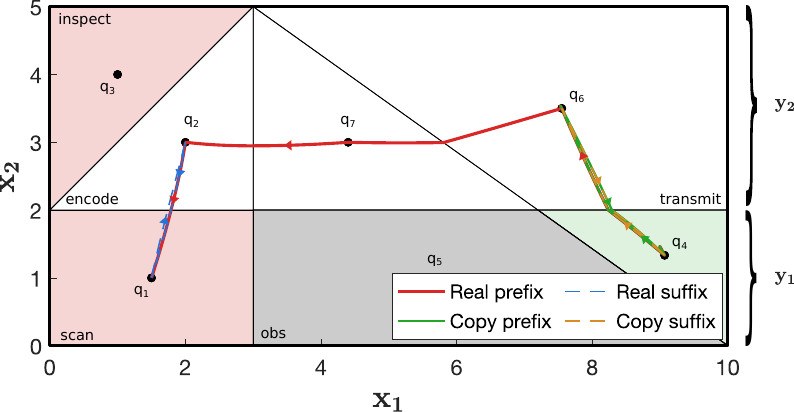}
        \centerline{(a) Drone~1 (real and copy trajectories)}
    \end{minipage}
    \vfill
    \begin{minipage}{0.95\linewidth}
        \includegraphics[width=\linewidth]{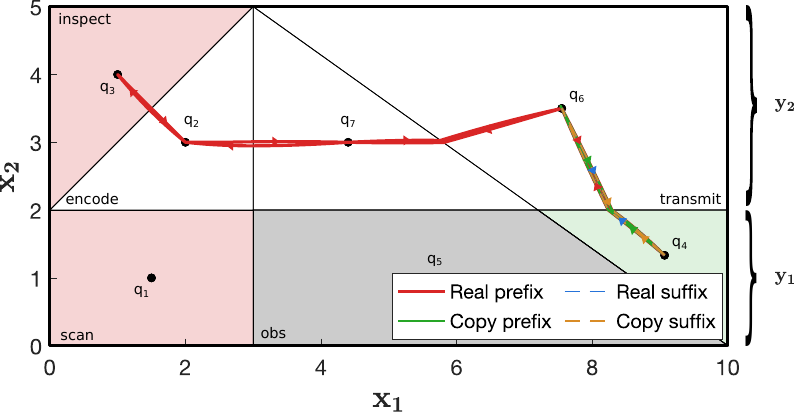}
        \centerline{(b) Drone~2 (real and copy trajectories)}
    \end{minipage}
    \caption{Workspace and synthesized trajectories. 
 }
    \label{fig:workspace_partition}
    \vspace{-0.15cm}
\end{figure}

We demonstrate our results in a two-drone scenario: Drone 1 scans air quality and encodes the collected data, while Drone 2 relays the data and visits a diagnostics station for inspection.
The bounded workspace $P = \{\mathbf{x} = \begin{bmatrix} x_1 & x_2 \end{bmatrix}^\top \, | \, 0 \leq x_1 \leq 10, \ 0 \leq x_2 \leq 5 \}$ is partitioned into seven convex sub-polytopes $q_1, \dots, q_7$ using a modified script from the LTLCon package \citep{kloetzer2007}. Each sub-polytope is labeled according to the atomic propositions $\mathcal{AP} = \{\texttt{scan}, \texttt{encode},$ $\texttt{transmit}, \texttt{inspect}, \texttt{obs} \}$ via the labeling function $L$. Two regions are secret: ~$q_1$ containing sensitive environmental data, and ~$q_3$, a secure diagnostics station. Both drones start from $q_4$ and must always avoid the obstacle region $q_5$, which models restricted airspace due to strong turbulence. An intruder monitors the airspace, which is divided into two observation regions: $y_1$ is the southern and $y_2$ is the northern section. The corresponding workspace partition is shown in Fig.~\ref{fig:workspace_partition}, where green, red, and gray areas denote initial, secret, and obstacle regions, respectively. 
The goal is to synthesize a trajectory satisfying Type A/B security and a mission specification given as an LTL formula:
\[
    \begin{aligned}
        \phi = \, &\square \, \lozenge \Big( \texttt{scan}^1 \land \lozenge \left( \texttt{encode}^1 \land \lozenge \texttt{transmit}^2 \right) \Big) \land\\ & \lozenge \texttt{inspect}^2 \land \square \, \lnot \, \texttt{obs}, 
\end{aligned}
\]
enforcing a repeating sequence: Drone~1 performs the air-quality scan, then encodes the data, followed by Drone~2's transmission. Moreover, Drone~2 must eventually execute an inspection, while both always avoid the obstacle. The superscript denotes the agent responsible for each proposition.

Each Drone~$i \in \{1,2\}$ follows linear affine dynamics \vspace{-0.2em}
\[
    \mathbf{\dot x}_i = \begin{bmatrix} 0.1 & -0.15 \\ -0.1 & 0.2 \end{bmatrix} \, \mathbf{x}_i + \begin{bmatrix} 1 & 0 \\ 0 & 1 \end{bmatrix} \, \mathbf{u}_i + \begin{bmatrix} 0.2 \\ 0.3 \end{bmatrix},\vspace{-0.3em}
\]where $\mathbf{x}_i = \begin{bmatrix} x_{i, 1} & x_{i, 2} \end{bmatrix}^\top \in P$, $\mathbf{u}_i = \begin{bmatrix} u_{i, 1} & u_{i, 2} \end{bmatrix}^\top \in U_i$ and $U_i = [-3, 3] \times [-3,3], \, \forall i \in \{1,2\}$.

The implementation was carried out in Python 3.11 on a MacBook Air M2. Polytope operations were performed using the \texttt{CDD} library \citep{cdd}. First, the discrete abstraction $\mathcal{T}_g$ of the system was constructed (Section~\ref{subsec:wts}). The Type-B gWTS was then obtained by removing states that violate Type-B security, while the Twin-gWTS and Secure Twin-gWTS were subsequently derived to retain only states satisfying Type-A security (Section~\ref{subsec:secure_gwts}). 
Transition feasibility in the resulted WTS was checked by solving the QP in Python, subject to the agent's dynamics, control bounds, facet-crossing constraints, and the CBF constraints (Section~\ref{subsec:feasibility}). For each transition, one CBF constraint is enforced for every facet of the current polytope over the interval $[0,t_c]$ and similarly for each facet of the next polytope over $[t_c,t_f]$. As an example, consider the transition in which Drone~$1$ moves from $q_2$ to $q_7$, and, specifically, the facet of $q_2$ given by $\begin{bmatrix} 1 & -1 \end{bmatrix} \mathbf{x} + 2 = 0$ with the CBF constraint:
$ \begin{bmatrix} 1 & -1 \end{bmatrix} (\mathbf{A}_1 \, \mathbf{x}_1 + \mathbf{B}_1 \, \mathbf{u}_1 + \mathbf{b}_1) \geq -\gamma \left( \begin{bmatrix} 1 & -1 \end{bmatrix} \, \mathbf{x}_1 + 2 \right) + \epsilon$ 
for all $t \in [0, t_c]$. We used the constants $\gamma=1$ and $\epsilon = 0.01$. Next, we follow Alg. \ref{alg:secure_optimal_trajectory_planning} to generate the continuous trajectory. 
The NBA corresponding to $\phi$ was generated using \texttt{LTL2BA} \citep{ltl2ba}. The secure abstraction was then synchronized with this NBA to construct the PBA (Section~\ref{subsec:task_satisfaction}). The shortest prefix–suffix path minimizing $J = 0.5 \, J_{\text{prefix}} + 0.5 \, J_{\text{suffix}}$ was computed, and the continuous trajectories were obtained by concatenating the trajectory segments associated with the discrete transitions along the prefix–suffix path (Section~\ref{subsec:secure_planning_alg}). 
The synthesized discrete path is: $(q_4, q_4) \to_g (q_4, q_6) \to_g (q_4, q_7) \to_g (q_4, q_2) \to_g (q_6, q_3) \to_g (q_6, q_2) \to_g (q_7, q_7) \to_g (q_2, q_6) \to_g [(q_1, q_4) \to_g (q_2, q_6)]^\omega  $, and the copy path: $(q_4, q_4) \to_g (q_4, q_6) \to_g (q_4, q_6) \to_g (q_4, q_6) \to_g (q_6, q_6) \to_g (q_6, q_6) \to_g (q_6, q_6) \to_g (q_6, q_6) \to_g [(q_4, q_4) \to_g (q_6, q_6)]^\omega$.  The runtime for the complete algorithm is approximately 70 sec on the MacBook Air M2. 

The corresponding continuous trajectories, shown in Fig.~\ref{fig:workspace_partition}, demonstrate that the LTL specification and the required security properties are satisfied. When Drone~1 visits the secret region $q_1$, the copy path simultaneously visits $q_4$, a non-secret region with the same observation, satisfying Type-A security, while Drone~2 is also in $q_4$, preventing the intruder from inferring which agent entered a secret state and ensuring Type-B security. Similarly, when Drone~2 visits $q_3$, its copy path is in $q_6$ and Drone~1 in $q_6$, satisfying Type-A/B security. These results demonstrate that the proposed approach generates dynamically feasible, LTL-satisfying, and secure trajectories for multi-agent systems.

\vspace{-0.2cm}
\section{Conclusion} \label{sec:conclusion}

In this paper, we addressed the problem of security-aware multi-agent motion planning and control under LTL specifications. We proposed a centralized framework that first computes a discrete path and then refines it into continuous, dynamically feasible trajectories that satisfy the LTL task while enforcing Type-A, Type-B, or Type-A/B security. 
Future work could focus on extending the approach to nonlinear dynamics, and investigating hierarchical architectures to improve scalability for multi-agent systems.
\vspace{-0.2cm}

\bibliography{ifacconf}                                     

@inproceedings{ames2019,
    author={A. Ames and S. Coogan and M. Egerstedt and G. Notomista and K. Sreenath and P. Tabuada},
    title={Control barrier functions: theory and applications},
    booktitle={Proc. 18th IEEE Eur. Conf. Control},
    year={2019},
    pages={3420--3431},
}

@book{baier2008,
    author={C. Baier and J. Katoen},
    title={Principles of Model Checking},
    publisher={MIT Press},
    year={2008},
}

@article{fainekos2009,
    author={G. E. Fainekos and A. Girard and H. Kress-Gazit and G. J. Pappas},
    title={Temporal logic motion planning for dynamic robots},
    journal={Automatica},
    year={2009},
    volume={45},
    number={2},
    pages={343--352},
}

@misc{cdd,
    author={K. Fukuda},
    title={cddlib: An efficient implementation of the Double Description Method},
    year={2025},
    address={ETH Z{\"u}rich, Switzerland},
}

@misc{ltl2ba,
    author={P. Gastin and D. Oddoux},
    title={\textsc{LTL2BA} (Version 1.3)},
    year={2020},
}

@article{guo2014,
    author={M. Guo and D. V. Dimarogonas},
    title={Multi-agent plan reconfiguration under local LTL specifications},
    journal={Int. J. Robot. Res.},
    year={2014},
    volume={34},
    number={2},
    pages={218--235},
}

@article{hadjicostis2018,
    author={C. N. Hadjicostis},
    title={Trajectory planning under current-state opacity constraints},
    journal={IFAC-PapersOnLine},
    year={2018},
    volume={51},
    pages={337--342},
}

@article{kloetzer2007,
    author={M. Kloetzer and C. Belta},
    title={Temporal logic planning and control of robotic swarms by hierarchical abstractions},
    journal={IEEE Trans. Robot.},
    year={2007},
    volume={23},
    number={2},
    pages={320--330},
}

@article{kloetzer2008,
    author={M. Kloetzer and C. Belta},
    title={A fully automated framework for control of linear systems from temporal logic specifications},
    journal={IEEE Trans. Autom. Control},
    year={2008},
    volume={53},
    number={1},
    pages={287--297},
}

@book{lavalle2006,
    author={S. M. LaValle},
    title={Planning Algorithms},
    publisher={Cambridge Univ. Press},
    year={2006},
}

@article{liu2022,
    author={S. Liu and A. Trivedi and X. Yin and M. Zamani},
    title={Secure-by-construction synthesis of cyber-physical systems},
    journal={Annu. Rev. Control},
    year={2022},
    volume={53},
    pages={30--50},
}

@inproceedings{loizou2004,
    author={S. G. Loizou and K. J. Kyriakopoulos},
    title={Automatic synthesis of multi-agent motion tasks based on LTL specifications},
    booktitle={Proc. 43rd IEEE Conf. Decis. Control},
    volume={1},
    year={2004},
    pages={153--158},
}

@inproceedings{saboori2007,
    author={A. Saboori and C. N. Hadjicostis},
    title={Notions of security and opacity in discrete event systems},
    booktitle={Proc. 46th IEEE Conf. Decis. Control},
    year={2007},
    pages={5056--5061},
}

@article{smith2011,
    author={S. L. Smith and J. Tumova and C. Belta and D. Rus},
    title={Optimal path planning for surveillance with temporal-logic constraints},
    journal={Int. J. Robot. Res.},
    year={2011},
    volume={30},
    number={14},
    pages={1695--1708},
}

@inproceedings{wolper1983,
    author={P. Wolper and M. Y. Vardi and A. P. Sistla},
    title={Reasoning about infinite computation paths},
    booktitle={Proc. 24th Annu. Symp. Found. Comput. Sci.},
    year={1983},
    pages={185--194},
}

@inproceedings{yang2020,
    author={S. Yang and X. Yin and S. Li and M. Zamani},
    title={Secure-by-construction optimal path planning for linear temporal logic tasks},
    booktitle={Proc. 59th IEEE Conf. Decis. Control},
    year={2020},
    pages={4460--4466},
}

@article{yin2021,
    author={X. Yin and M. Zamani and S. Liu},
    title={On approximate opacity of cyber-physical systems},
    journal={IEEE Trans. Autom. Control},
    year={2021},
    volume={66},
    number={4},
    pages={1630--1645},
}

@article{yu2022,
    author={X. Yu and X. Yin and S. Li and Z. Li},
    title={Security-preserving multi-agent coordination for complex temporal logic tasks},
    journal={Control Eng. Pract.},
    year={2022},
    volume={123},
    pages={105130},
}

@article{zhong2025,
    author={B. Zhong and S. Liu and M. Caccamo and M. Zamani},
    title={Secure-by-construction synthesis for control systems},
    journal={IEEE Trans. Autom. Control},
    year={2025},
    volume={70},
    number={6},
    pages={4170--4177},
}

@article{ma2021optimal,
  title={Optimal secret protections in discrete-event systems},
  author={Ma, Ziyue and Cai, Kai},
  journal={IEEE Trans. Autom. Control},
  volume={67},
  number={6},
  pages={2816--2828},
  year={2021},
}

@inproceedings{wang2020hyperproperties,
  title={Hyperproperties for robotics: Planning via HyperLTL},
  author={Wang, Yu and Nalluri, Siddhartha and Pajic, Miroslav},
  booktitle={IEEE Int. Conf. Robot. Autom.},
  pages={8462--8468},
  year={2020}
}

@inproceedings{xie2021secure,
  title={Secure-by-construction controller synthesis for stochastic systems under linear temporal logic specifications},
  author={Xie, Yifan and Yin, Xiang and Li, Shaoyuan and Zamani, Majid},
  booktitle={60th IEEE Conf. Decis. Control},
  pages={7015--7021},
  year={2021}
}

\end{document}